\documentclass[letterpaper,twocolumn,10pt]{article}
\usepackage{usenix-2020-09}

\usepackage{tikz}
\usepackage{float}
\usepackage{amsmath}
\usepackage{balance}
\usepackage{booktabs}
\usepackage[capitalise]{cleveref}

\crefformat{section}{#2\S#1#3}
\crefname{figure}{Fig.}{Figs.}
\crefname{table}{Tab.}{Tabs.}
\usepackage{colortbl}
\usepackage{xcolor}
\usepackage{xspace}
\usepackage{multirow}
\usepackage{tabularx}
\usepackage{array}
\usepackage[font=small,labelfont={small,bf}]{caption}
\usepackage[font=small,labelfont={small,bf}]{subcaption}
\usepackage{graphicx}
\usepackage{thm-restate}
\usepackage{tablefootnote}
\usepackage{fontawesome}
\usepackage{hyperref}

\usepackage{enumitem}
\setlist[itemize]{noitemsep, topsep=0pt, leftmargin=*}

\usepackage{amssymb}
\usepackage{amsthm}
\usepackage{mathtools}
\usepackage{subcaption}

\DeclareMathOperator*{\argmin}{arg\,min}
\usepackage{algorithm}
\usepackage{algpseudocode}
\algtext*{EndWhile} 
\algtext*{EndIf} 
\algtext*{EndFor} 

\newcommand{\ie}{i.e.\xspace}
\newcommand{\eg}{e.g.\xspace}

\theoremstyle{plain}
\usepackage{thmtools}
\newtheorem{theorem}{Theorem}[section]

\theoremstyle{definition}
\newtheorem{definition}[theorem]{Definition}

\theoremstyle{remark}

\makeatletter
\renewcommand{\paragraph}{\smallskip\noindent\textbf}
\makeatother

\long\def\ignore#1{}

\def\Snospace~{\S{}}

\newcommand{\distname}{D$^2$LPM\xspace}

\begin{document}

\date{}

\title{\Large \bf Locality-aware Fair Scheduling in LLM Serving }
\author{Shiyi Cao$^{*\, 1\, 2}$, Yichuan Wang$^{*\, 1\, 2}$, Ziming Mao$^{1}$, Pin-Lun Hsu$^{2}$, Liangsheng Yin$^{1\,2}$, Tian Xia$^{1}$, Dacheng Li$^{1\,2}$ \\ Shu Liu$^{1}$, Yineng Zhang$^{2}$, Yang Zhou$^{1}$, Ying Sheng$^{2}$, Joseph Gonzalez$^{1}$, Ion Stoica$^{1}$ }
\maketitle
{\let\thefootnote\relax\footnote{$^*$indicates equal contribution. \textsuperscript{1}UC Berkeley \textsuperscript{2}LMSYS 
}

\begin{abstract}
Large language model (LLM) inference workload dominates a wide variety of modern AI applications, ranging from multi-turn conversation to document analysis. Balancing fairness and efficiency is critical for managing diverse client workloads with varying prefix patterns. Unfortunately, existing fair scheduling algorithms for LLM serving, such as Virtual Token Counter (VTC), fail to take prefix locality into consideration and thus suffer from poor performance. On the other hand, locality-aware scheduling algorithms in existing LLM serving frameworks tend to maximize the prefix cache hit rate without considering fair sharing among clients. 

This paper introduces the \emph{first} locality-aware fair scheduling algorithm, Deficit Longest Prefix Match (DLPM), which can maintain a high degree of prefix locality with a fairness guarantee. We also introduce a novel algorithm, Double Deficit LPM (\distname), extending DLPM for distributed setup that can find a balance point among fairness, locality, and load-balancing. Our extensive evaluation demonstrates the superior performance of DLPM and \distname in ensuring fairness while maintaining high throughput (up to 2.87$\times$ higher than VTC) and low per-client (up to 7.18$\times$ lower than state-of-the-art distributed LLM serving system) latency.
\end{abstract}
\section{Introduction}
Online inference workloads for large language models (LLMs) are rapidly becoming widespread, driven by their general-purpose capabilities and versatility across a wide range of tasks such as search engines~\cite{perplexity}, coding assistant~\cite{copilot}, autonomous agents~\cite{Park2023GenerativeAgents,wang2023voyager,openai2023gpt4}, and tool calling~\cite{schick2023toolformer,patil2023gorilla}. The release of OpenAI's o1 model has further highlighted the \emph{test-time scaling} phenomenon~\cite{brown2024large,openai_learning_to_reason_2024, snell2024scaling,deepseek_r1_lite_2024}, where the allocation of increased computational resources during inference via techniques such as Monte Carlo Tree Search (MCTS)~\cite{zhang2024rest,putta2024agent}, Best-of-N sampling~\cite{snell2024scaling} and Self-refine~\cite{madaan2024self}, can substantially improve the quality of LLM-generated answers across various tasks.
The increasingly complex test-time compute requirements underscore the growing prominence of inference workloads in the LLM landscape.

Despite the advance in LLM generation quality, efficiently scaling online LLM inference services remains challenging, posing substantial barriers to their broad adoption. On the one hand, service providers need to provide isolation between concurrent tasks to ensure stable and predictable performance for all clients~\cite{openai2023rate}: a client's experience should not be negatively impacted by a dominant or malicious client. On the other hand, service providers want to maximize system efficiency to improve throughput and reduce cost.

\begin{figure}[t]
\centering
\includegraphics[width=0.34\textwidth]{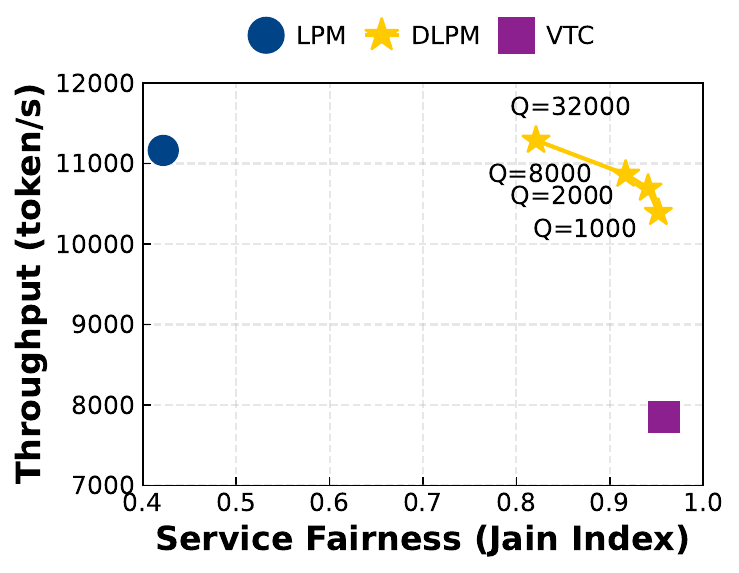}
\caption{DLPM achieves a \emph{new Pareto frontier} considering locality
and fairness in LLM serving. $Q$ is a hyper-parameter in DLPM, indicating how much we relax the fairness bound of DLPM. Results are obtained in a single A10 GPU. }
\label{fig:pareto}
\vspace{-7mm}
\end{figure}
Unfortunately, existing scheduling algorithms~\cite{kwon2023efficient,vtc,zheng2023efficiently,preble} for LLM serving fall short of achieving these dual goals effectively, as shown in \cref{fig:pareto}.
Although fair scheduling algorithms such as Virtual Token Counter (VTC)~\cite{vtc}, are work-conserving -- ensuring the system is fully utilized as long as there are requests in the system -- they are not locality-aware.
Locality awareness is essential for enhancing memory and computational efficiency, particularly through mechanisms such as prefix sharing~\cite{zheng2023efficiently}. Reusing the prefix’s key-value (KV) tensors across multiple requests allows multiple requests sharing the same prefix to retain only one copy of the prefix's KV tensors in GPU memory. Moreover, it reduces redundant computation of the prefix’s KV tensors. Conversely, algorithms such as Longest Prefix Match (LPM)~\cite{zheng2023efficiently} enhance the system efficiency by prioritizing prefix locality: reordering the requests to maximize the prefix cache hit rate, yet they fail to guarantee effective isolation among clients -- a malicious client can monopolize shared resources by sending a large volume of requests with long identical prefix, significantly degrading the performance experienced by other clients.

Achieving both fairness and prefix locality in LLM inference scheduling is challenging, as these two goals inherently conflict with each other.
Prefix sharing, for instance, may require reordering requests to group those with identical prefixes together. In contrast, fair scheduling algorithms prioritize serving requests in a specific order to ensure isolation and prevent any single client from dominating resources. This necessary ordering can interfere with the efficiency gains from prefix sharing, as it restricts the flexibility to reorder requests for optimal resource utilization. This challenge is exacerbated in a distributed setting, where the algorithm must decide not only the order in which the requests are dispatched, but also to which GPU they are dispatched to achieve load balancing and prefix locality. This dual consideration of dispatch order and location significantly complicates achieving efficient and fair resource allocation across multiple GPUs. 

In this paper, we introduce the first locality-aware fair scheduling algorithm \textbf{Deficit Longest Prefix Match (DLPM)} for LLM serving which relaxes the dispatch order required by VTC to better preserve prefix locality while still bounding the allocation fairness.
As illustrated in \cref{fig:pareto}, DLPM can achieve throughput comparable to that of LPM while maintaining a degree of fairness close to that provided by VTC. We further propose a novel \emph{distributed} scheduling algorithm \textbf{Double Deficit LPM
(\distname)}
that builds on top of DLPM to preserve high per-GPU prefix locality with a global fairness guarantee in a distributed setting.

In summary, this paper makes the following contributions: 
\begin{itemize}
    \item We introduce the \emph{first} locality-aware fair scheduling algorithm DLPM and its distributed version \distname  for LLM serving, which can achieve up to $2.87\times$ higher throughput than VTC and up to $7.18\times$ lower latency than the state-of-the-art locality-aware scheduling algorithm~\cite{zheng2023efficiently,preble}.
    \item We provide rigorous theoretical bounds on DLPM and \distname's fairness property, including service bound and latency bound between various types of clients.
    \item We conduct extensive evaluations on our proposed algorithms and demonstrate their superiority in achieving high system throughput while preserving fairness guarantees.
\end{itemize}
\section{Background and Motivation}
In this section, we first briefly introduce the basics of LLM inference, prefix caching, and fairness in LLM serving (\cref{sec:b1}). We then discuss key issues with existing LLM serving scheduling algorithms and the challenges they pose (\cref{sec:b2}).
\subsection{Transformer-Based LLM Inference}
\label{sec:b1}

\paragraph{LLM Inference}
Modern transformer-based LLM inference consists of \textit{prefill} and \textit{decode} phases.
The prefill phase takes input prompt, computes internal embedding vectors for all prompt tokens in parallel using the attention mechanism~\cite{vaswani2017attention}, and generates the first output token. 
These embedding vectors are normally stored inside the GPU memory as the \emph{KV cache} to avoid recomputation. 
In the decode phase, new tokens are generated auto-regressively until an End-Of-Sequence (EOS) token is encountered or the pre-defined maximum token length is reached. During each iteration of token generation, the key-value (KV) cache of all previous tokens will be needed and the key-value tensors of the newly generated token will be appended to the KV cache.
Such auto-regressive generation can lead to sub-optimal device utilization and decreased serving throughput~\cite{pope2023efficiently}. To enhance GPU utilization, \cite{yu2022orca} proposed \emph{continuous batching}. However, limited memory capacity emerged as a critical bottleneck, restricting batch sizes and thus reducing GPU efficiency. To address this issue, \cite{kwon2023efficient} developed PagedAttention, which mitigates memory fragmentation inherent in continuous batching and significantly enhances memory efficiency.

\paragraph{Prefix Caching and Locality}
To further improve the memory and computation efficiency, SGLang~\cite{zheng2023efficiently} introduced RadixAttention to facilitate the reuse of the KV cache of the shared prefix across multiple different LLM calls.
By exploiting the prefix locality, memory usage for the KV cache is reduced, allowing for larger batch sizes and improved GPU utilization. Additionally, it eliminates redundant computations for the shared KV cache,

This technique is increasingly crucial for emerging multi-call LLM workloads such as Tree-of-Thoughts~\cite{yao2024tree}, Skeleton-of-Thought~\cite{ning2024skeletonofthought}, MCTS~\cite{zhang2024rest}, and Self-Refine~\cite{madaan2024self}, where there are substantial opportunities for prefix sharing. For instance, in a Tree-of-Thoughts program, all branches originating from the same node share the entire prefix up to the root. As the tree expands, the number of requests sharing the same prefix grows, and as the tree deepens, the length of the shared prefix increases.

\paragraph{LLM Serving Fairness}
Achieving efficient online LLM inference with Service Level Objective (SLO) guarantees necessitates isolation among different clients~\cite{vtc}. 
This need arises because clients share the same GPU accelerators and compete for these GPU resources. Without isolation, there is a risk that one client might monopolize resources, leading to the starvation of others. 
Moreover, to optimize resource utilization, it is crucial to reallocate unused resources from one client to another rather than merely imposing a rate limit~\cite{openai2023rate} on each client for isolation purposes. 
Rate limits simply disallows clients to send requests beyond a certain rate which harms the resource utilization as shown in~\cite{vtc}.
Formally, our goal is to achieve the classic max-min fairness~\cite{maxminfair}, where the fair scheduling ensures each client receives at least 1/$n$ of the resources, with $n$ representing the total number of clients. If some clients do not fully utilize their allocated share, these resources can be redistributed to others.

The first fair scheduling algorithm targeting the continuous batching mechanism in online LLM serving was the virtual token counter (VTC)~\cite{vtc}. 
VTC maintains a virtual counter of the tokens serviced for each client and prioritizes clients with the lowest counters in each batching iteration. 
By tracking token-level resource usage, VTC achieves fair scheduling even when the output length of the request is unknown.

\subsection{The Trade-offs}
\label{sec:b2}

\begin{figure}[!t]
\centering
\begin{minipage}{0.15\textwidth}
    \centering
    \includegraphics[width=\linewidth]{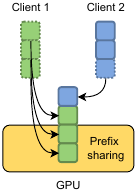}
    \subcaption{LPM.}
    \label{fig:locality_aware}
\end{minipage}
\hfill
\begin{minipage}{0.15\textwidth}
    \centering
    \includegraphics[width=\linewidth]{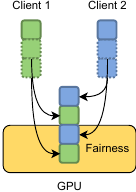}
    \subcaption{VTC.}
    \label{fig:vtc_fair}
\end{minipage}
\hfill
\begin{minipage}{0.15\textwidth}
    \centering
    \includegraphics[width=\linewidth]{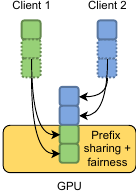}
    \subcaption{DLPM (ours).}
    \label{fig:dlpm_sched}
\end{minipage}
\caption{Requests from the same client share prefixes with each other. 
In LPM locality-aware scheduling, the system schedules the GPUs to process all requests from Client 1 to maximize prefix sharing while starving Client 2. In VTC fair scheduling, the system processes requests in turn to maximize fairness, while ignoring the prefix sharing opportunity.
Our DLPM scheduling achieves the best of two worlds through a novel quantum mechanism (\cref{sec:local_dlpm}) to guarantee locality while not sacrificing fairness.}
\label{fig:background}
\vspace{-1em}
\end{figure}

\paragraph{Locality vs. Fairness}
Achieving both strong fairness and high locality for efficient online LLM serving is inherently challenging, since these two are usually at odds with each other, as illustrated in \cref{fig:background}.
On the one hand, locality-aware scheduling (\cref{fig:locality_aware}) reorders requests to group those with similar prefixes -- often originating from the same client -- to the same GPU to optimize for prefix locality. 
On the other hand, the VTC fair scheduler (\cref{fig:vtc_fair}) adheres to a strict order based on per-client resource usage counters to dispatch requests, ensuring no client continuously dominates the GPU usage; such an order compromises locality as it intersperses the requests of the same client with requests from other clients.
\cref{fig:pareto} also demonstrates the vastly different prioritizations of these two techniques, highlighting the trade-off between fairness and prefix locality.

\paragraph{Locality vs. Load-Balancing}
The challenge intensifies in distributed settings, where model replicas are served on multiple workers, each managed by its own local scheduler, with a global scheduler coordinating all these local workers. In this scenario, the scheduling algorithm on the global scheduler must balance a trade-off between locality and load balancing. For instance, simply distributing requests equally across the cluster is suboptimal due to the high prefix recompute overhead. Similarly, always dispatching requests with the same prefix to a single GPU can lead to workload imbalance.
 
\begin{figure}[t]
\centering
\includegraphics[width=0.45\textwidth]{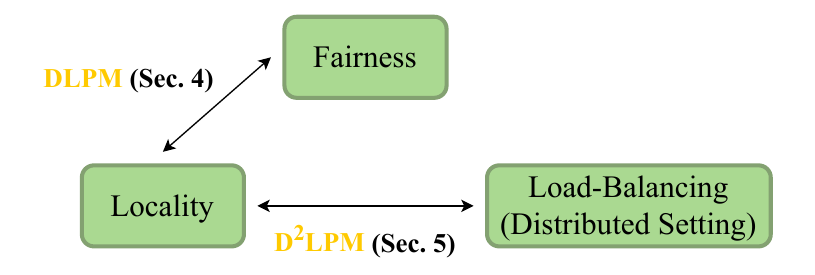}
\caption{This paper addresses the conflict between fairness and locality through the DLPM mechanism (\cref{sec:local_dlpm}). It further addresses the conflict between locality and load balancing in distributed settings with the \distname mechanism (\cref{sec:dist_dlpm}).}  
\label{fig:trinity}
\end{figure}

\paragraph{Design Goals} The main goal of this paper is to provide a principled way of navigating the trade-off between strong fairness and high locality in online LLM serving, as well as between locality and load-balancing in distributed settings. 
Our methodology ensures that the algorithms for single and distributed settings can be combined to maintain global fairness effectively.
In the remainder of the paper, we begin by discussing preliminary concepts related to fairness in LLM serving (\cref{sec:preliminary}), then we introduce our fair scheduling design for a single worker (\cref{sec:local_dlpm}), and finally, we expand this approach to distributed fair scheduling (\cref{sec:dist_dlpm}).
\section{Preliminaries}
\label{sec:preliminary}
In this section, we first formally define the properties a fair scheduling algorithm needs to meet for LLM serving, following those described in VTC~\cite{vtc}. We then discuss the cost function we adopt for service measurement.

\begin{definition}[Backlog]
A client $u$ is backlogged if dispatching additional requests cannot further increase throughput and can only incur additional queueing delay. In distributed settings, a backlogged client may have requests in queues of certain workers or all workers (depending on the policy).
\label{def:backlog}
\end{definition}

\paragraph{Fairness Properties} Similar to VTC, our goal is to achieve \textit{approximate max-min fairness}~\cite{maxminfair} on the service received by each client; \textbf{different from} VTC, we also want to preserve prefix cache locality. More formally, an LLM serving system that can achieve \emph{approximate} max-min fairness should satisfy the following three properties~\cite{vtc}:
\begin{enumerate}
    \item During any time interval $[t_1, t_2)$, if two clients $f$ and $g$ are continuously backlogged, they should receive a similar level of service, i.e. $\vert W_f(t_1, t_2) - W_g(t_1, t_2) \vert \leq \delta$, where $\delta$ is a constant value independent of $t_2 - t_1$.
    \item A client $f$ that is continuously backlogged during a time interval should not receive less service than another client $g$ that is not continuously backlogged during the same time interval, i.e. $ W_g(t_1, t_2) - W_f(t_1, t_2) \leq \delta$, where $\delta$ is a constant value.
    \item The scheduling policy should be work-conserving: no worker should be idle if there are requests in the queue.
\end{enumerate}

The first property states that a client sending at a high request rate is guaranteed to not receive more than their fair share of service and will not impact other normal-behaved clients. The second property prevents clients from accumulating unused service by first sending at a low request rate and later monopolizing the system.
The third property guarantees that no resources are wasted in order to enforce fairness.

\paragraph{Measurement of Service}
Another important aspect in designing a fair scheduling algorithm for LLM serving is how the service should be measured. In VTC, the cost function is defined as a weighted sum of the number of input tokens and the number of output tokens. To incorporate the impact of prefix sharing (\ie, reduced memory and computations), we introduce a slightly different measure.
Intuitively, with prefix sharing, the prefix tokens' cost should only be counted once when it is first calculated and stored in the GPU memory. Our prefix-aware version of the cost function is then defined as $W(t_1, t_2) = w_e \cdot n_e(t_1, t_2) + w_q \cdot n_q(t_1, t_2)$. The notations are explained in \cref{tab:notations}. Here $w_e$ and $w_q$ are set to be $1$ and $2$, inspired by OpenAI's pricing for GPT4\footnote{\url{https://openai.com/api/pricing/}}.

\begin{table}[!t]
\centering
\caption{The upper half includes notations for service measurement. The lower half includes notations for the DLPM and \distname algorithm and their analysis. *The extend tokens are the input tokens excluding prefix tokens.}
\resizebox{1\columnwidth}{!}{
\begin{tabular}{c|p{73mm}}
\toprule
Notation & Explanation\\
\midrule
$W_f(t_1, t_2)$ & service received by $f$ during interval $[t_1, t_2)$\\
$n_e$ & number of processed extend tokens* \\
$n_q$ & number of processed output tokens \\
$w_e$ & weight of extend tokens in the cost function \\
$w_q$ & weight of output tokens in the cost function \\
\midrule
$Q^u$ & the quantum assigned to each client in DLPM\\
$q_i$ & the deficit counter of client $i$ in DLPM \\
$Q^w$ & the quantum assigned to each worker in \distname\\
$q_{i,w}$ & the deficit counter of worker $w$ for client $i$ in \distname \\
$L_{input}$ & maximum number of input tokens in a request \\
$L_{output}$ & maximum number of output tokens in a request \\
$M$ & maximum number of tokens that can be fitted in a running batch  \\
$U$ & maximum number of counter that a single request can consume $w_e\cdot L_{input} + w_q\cdot M$ \\
$D$   & data parallelism degrees \\
\bottomrule
\end{tabular}
}
\label{tab:notations}
\end{table}
\section{Deficit Longest Prefix Match (DLPM)}
\label{sec:local_dlpm}
In this section, we present our algorithm DLPM for the single worker in \cref{sec:dlpm_alg} and the proved fairness guarantees in \cref{sec:dlpm_property}.

\subsection{Algorithm Design}\label{sec:dlpm_alg}
In the Longest Prefix Match (LPM) algorithm~\cite{zheng2023efficiently}, at each continuous batching step, the scheduler first sorts current requests in the waiting queue based on their matched prefix length and then adds them to the new batch until the memory pool is full. LPM efficiently utilizes memory by grouping requests that can share a common prefix, thus maximizing the decoding batch size, which in turn leads to better operational intensity and throughput for the decoding phase.
To maintain the cache hit rate while introducing a fairness guarantee, it is essential not to disrupt the LPM order of the requests excessively. To achieve this, we incorporate a quantum mechanism inspired by the deficit round robin (DRR) approach~\cite{drr}. This mechanism compels the scheduler to occasionally prioritize requests from less-served clients over those with the longest matching prefixes. Intuitively, this mechanism is effective because it preserves the local ordering inherent to the LPM. As a result, the system continues to benefit significantly from the memory savings brought by the shared prefixes, while the additional cost of prefix recomputation is incurred only when switching to serve less-served clients.
This balanced approach allows DLPM to uphold the core efficiencies of the original LPM algorithm while enhancing fairness across client requests, ensuring that no clients monopolize the batching process to the detriment of others. 

The core algorithm of DLPM is presented in \cref{alg:local-dlpm}. Initially, the algorithm initializes all clients' deficit counter \(q_i\) to zero, with \(Q^u\) representing the service quantum replenished to each client in a cycle. At each continuous batching step, DLPM performs the following steps: 1) It sorts the requests in the waiting queue by their matched prefix length and then tries to add them to the currently running batch (\(B\)) until the memory pool is full. 2) The request will be added to \(B\) if the corresponding client has a positive deficit counter (\(q_i>0\)). Otherwise, the request will be skipped. When all the active clients have $q\leq0$, they will be replenished by $Q^u$ at Line~\ref{line:replenish_c}. 3) After each request is added to \(B\), the corresponding client's deficit counter will deduct the amount of service invoked by the extend tokens. 4) The new batch \(B\) then goes through one model forward step. After each decoding step, the service invoked by the output tokens will be deducted from the client's deficit counter accordingly.

\begin{algorithm}[!t]
\caption{Deficit Longest Prefix Match (DLPM)}
\begin{algorithmic}[1]
\State let $l$ denotes the client list
\State let $B$ denotes current running batch
\Function{CheckRefill}{$l$, $Queue$}
\ForAll{$i \in \{ client(r) \mid r \in Queue \}$}
    \If{$q_i > 0$} \label{line:check_q_i}
        \Return
    \EndIf
\EndFor
\ForAll{$i \in l$}
    \If{$q_i \leq 0$}
        $q_i\leftarrow q_i + Q^u$
        \label{line:replenish_c}
    \EndIf
\EndFor
\EndFunction
\State $\triangleright$ \texttt{with monitoring stream:}
\While{True}
    \If {new request $r$ from client $i$ arrived}
        \If {$i \notin l$}
            $q_i\leftarrow 0$, \label{line:client_join}
            $l \leftarrow l + u$
        \EndIf
        \State $Queue \leftarrow Queue + r$
    \EndIf
\EndWhile
\State $\triangleright$ \texttt{with execution stream 1:}
\While{True}
    \State $Queue \leftarrow $ \Call{SortByPrefix}{$Queue$} 
    \While{not $Queue.empty()$}
        \For{each $r \in Queue$}
            \State $i \leftarrow client(r)$
            \If{$q_i \leq 0$}
                \Call{CheckRefill}{$l$, $Queue$}
            \EndIf
            
            \If{$q_i > 0$ and \Call{CanAdd}{$r$}}
                \State $B \leftarrow B + r$
                \State $q_i \leftarrow q_i - w_e\cdot extend\_length(r)$
                \label{line:subtract_c}
                \State $Queue \leftarrow Queue - r$
            \EndIf
        \EndFor
    \EndWhile
    \State \Call{ForwardStep}{$B$}
    \State $q_i\leftarrow q_i - w_q \cdot \vert \{r| client(r)=i, r \in B\}\vert$
    \label{line:subtract_c_finish}
    \State $B \leftarrow$ filter\_finished\_requests($B$)
\EndWhile
\end{algorithmic}
\label{alg:local-dlpm}
\end{algorithm}

\subsection{Fairness Guarantees of DLPM}\label{sec:dlpm_property}
In this section, we provide the theoretical fairness guarantees of DLPM that correspond to the three properties we introduced in \cref{sec:preliminary}. The full proofs are provided in \cref{sec:appendix_dlpm}.

\begin{restatable}[\textbf{Service bound between backlogged clients}]{theorem}{boundedServiceDifference} Under the DLPM scheme: for any time interval $[t_1,t_2)$, if two clients $f$ and $g$ are continuously backlogged. Then the difference in their received service are bounded: $\vert W_f(t_1,t_2) - W_g(t_1,t_2)\vert \leq 2 \cdot (U + Q^u)$, where $U=w_e\cdot L_{input} + w_q\cdot M$.
\label{theorem:bounded-service-difference}
\end{restatable}

\begin{proof}
Let the client with maximum service be $f$, and the client with minimum service be $g$. Consider $t_1$ and $t_2$.
\begin{itemize}
    \item At $t_2$, since both clients $f$ and $g$ are backlogged and are in client list $l$, both client $f$ and client $g$ have been replenished the \textit{same} $k$ number of times in Line~\ref{line:replenish_c} since $t_1$.  $f$ and $g$ are backlogged, Line~\ref{line:check_q_i} ensures that both clients have negative $q_i$ before reaching Line~\ref{line:replenish_c} and be replenished.
    \item Since $t_1$, client $f$ at $t_2$ has received service $W_f(t_1, t_2) =  q_f(t_1) + k \cdot Q^u - q_f(t_2)$. client $g$ at $t_2$ has received service $W_g(t_1, t_2) =  q_g(t_1) + k \cdot Q^u - q_g(t_2)$.
    \item  $\vert W_f(t_1, t_2) - W_g(t_1, t_2) \vert = \vert q_f(t_1) - q_f(t_2) - q_g(t_1) + q_g(t_2)  \vert \leq \vert q_f(t_1) - q_f(t_2) \vert + \vert q_g(t_2) - q_g(t_1) \vert \leq 2 \cdot (U + Q^u)$, according to \cref{theorem:dlpmServiceBound}.
\end{itemize}

\end{proof}

\begin{restatable}[\textbf{Service bound between backlogged and non-backlogged clients}]{theorem}{nonBackloggedClients}
\label{theorem:backlog}
Under the DLPM scheme: Client $f$ that is continuously backlogged during time interval $[t_1,t_2)$ should not receive less service than another client, $g$, that is not continuously backlogged during the same time interval, that is $W_f(t_1,t_2) \geq W_g(t_1,t_2) - 2 U - 2 Q^u$.
\end{restatable}

\begin{proof}

\begin{itemize}
    \item Consider client $f$ and client $g$. $f$ is \textit{continuously} backlogged and $g$ is not \textit{continuously} backlogged.
    \item If $g$ is not backlogged during the entire duration from $t_1$ to $t_2$, $W_g(t_1,t_2) \leq U$, with no new request arrival.
    \item Let client $f$ be replenished $k^{t}_f$ at time $t$ in Line~\ref{line:replenish_c}.
    \item Since $f$ is continuously backlogged from $t_1$ to $t_2$, $k^{t_2}_f - k^{t_1}_f \geq k^{t_2}_g - k^{t_1}_g$. A backlogged client will be replenished for the same time as another backlogged client, from Theorem~\ref{theorem:bounded-service-difference}. A non-backlogged client will be replenished less as it is not in the active client list (Line~\ref{line:check_q_i}). 
    \item $W_g(t_1, t_2) - W_f(t_1,t_2) = (q_g(t_1)  + k_g^{t_2} Q^u - q_g(t_2) - k_g^{t_1} Q^u) - ( q_f(t_1) + k_f^{t_2} Q^u -  q_f(t_2) + k_f^{t_1} Q^u) \leq 2 (U + Q^u) - Q^u \cdot (k_f^{t_2}  - k_f^{t_1} - k_g^{t_2} + k_g^{t_1}) \leq 2 (U + Q^u)$, since $Q^u \cdot (k_f^{t_2}  - k_f^{t_1} - k_g^{t_2} + k_g^{t_1}) > 0$.
\end{itemize}

\end{proof}

The DLPM algorithm is work-conserving since it only manipulates the dispatch order and does not reject a request if it fits into the running batch.

\cref{theorem:bounded-service-difference} and \cref{theorem:backlog} reflect the first and second properties introduced in \cref{sec:preliminary}. Illustrative examples for \cref{theorem:bounded-service-difference} can be found in \cref{fig:multi-turn-exp} and \cref{fig:mix} in \cref{sec:tot_vis}, where within any time interval, the difference of the received service of two continuously backlogged clients is bounded.

\section{Applying DLPM to Distributed Scheduling}\label{sec:dist_dlpm}
In this section, we first present the strawman solution of centralized DLPM for distributed scheduling that ignores the scheduling overhead (\cref{ssec:centra_dlpm}).
We then proposed a decentralized DLPM solution that hides this overhead while preserving fairness property (\cref{ssec:decen_dlpm}).

\subsection{Strawman: Centralized DLPM}
\label{ssec:centra_dlpm}
The DLPM algorithm works perfectly when there is no scheduling overhead such that the DLPM scheduler could immediately make decisions based on freshest GPU states. 
Unfortunately, in real-world distributed scenarios, scheduling overhead happens significantly because of concurrent request handling and synchronization, prefix tree traversing and maintenance, and more. 
Recent work has also shown that the CPU scheduling overhead occupies nearly half of the inference time for two popular LLM inference engines~\cite{scheduler_overhead}.

\paragraph{Global-local States Synchronization}
To enable global DLPM for fair scheduling in distributed setups, we need to synchronize local and global prefix caching information. This synchronization ensures that the global scheduler can replicate the decision-making process typical of a single worker.
Using the token RadixTree from SGLang~\cite{zheng2023efficiently} as an example, to construct an accurate global RadixTree at time \(t_i\) (assume the last time the global scheduler dispatches the requests at time \(t_{i-1}\)), updates from each worker $s$ are encapsulated as \( \Delta Tree_s \), defined as:
\[
\Delta Tree_s(t_{i-1}, t_i) = (N_{\text{inserted}}, N_{\text{evicted}}, M_{\text{KV}})_s
\]
where \( N_{\text{inserted}} \) and \( N_{\text{evicted}} \) are sets of nodes that have been inserted to or evicted from the RadixTree, between the last dispatch time \( t_{i-1} \) and the current time \( t_i \). \( M_{\text{KV}} \) indicates the current available KV cache memory.

Upon sending these updates, the worker enters a blocked state, awaiting new requests from the global scheduler. The global scheduler then updates the RadixTree accordingly and dispatches new requests to the local worker following the DLPM algorithm. Such a synchronous approach guarantees the effectiveness and correctness of DLPM in the distributed setup; however, it incurs significant overhead due to the need to block workers while awaiting new requests, and the race conditions on the global waiting queue across workers.

\begin{figure}[!t]
\subfloat[Scheduler Overhead Breakdown. The global queue size is 200. Decode batch size is 25.]{\includegraphics[width=0.275\textwidth]{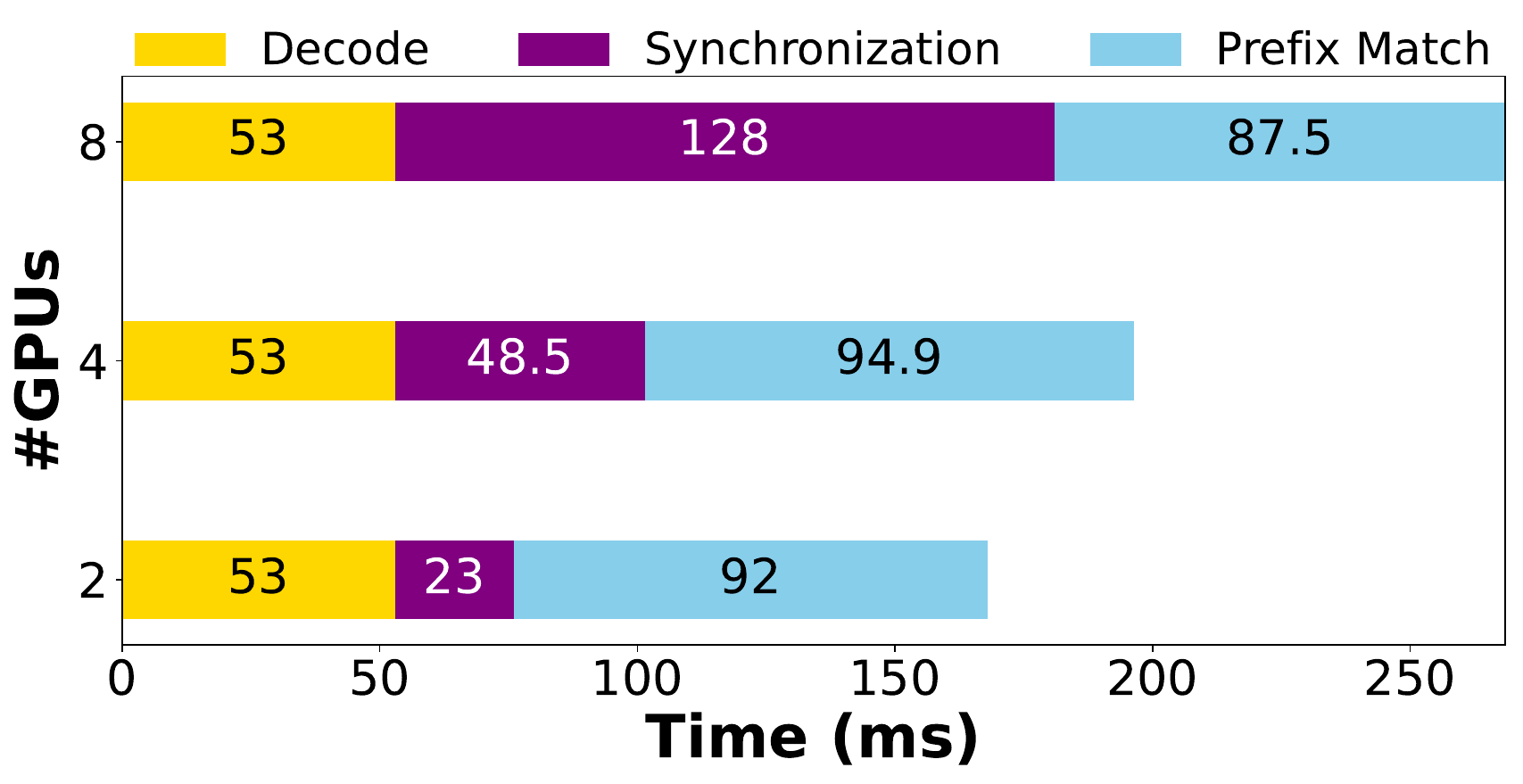}\label{fig:breakdown}}
\hspace{0.5em}
\subfloat[Prefix match overhead w.r.t global queue size.]{\includegraphics[width=0.18\textwidth]{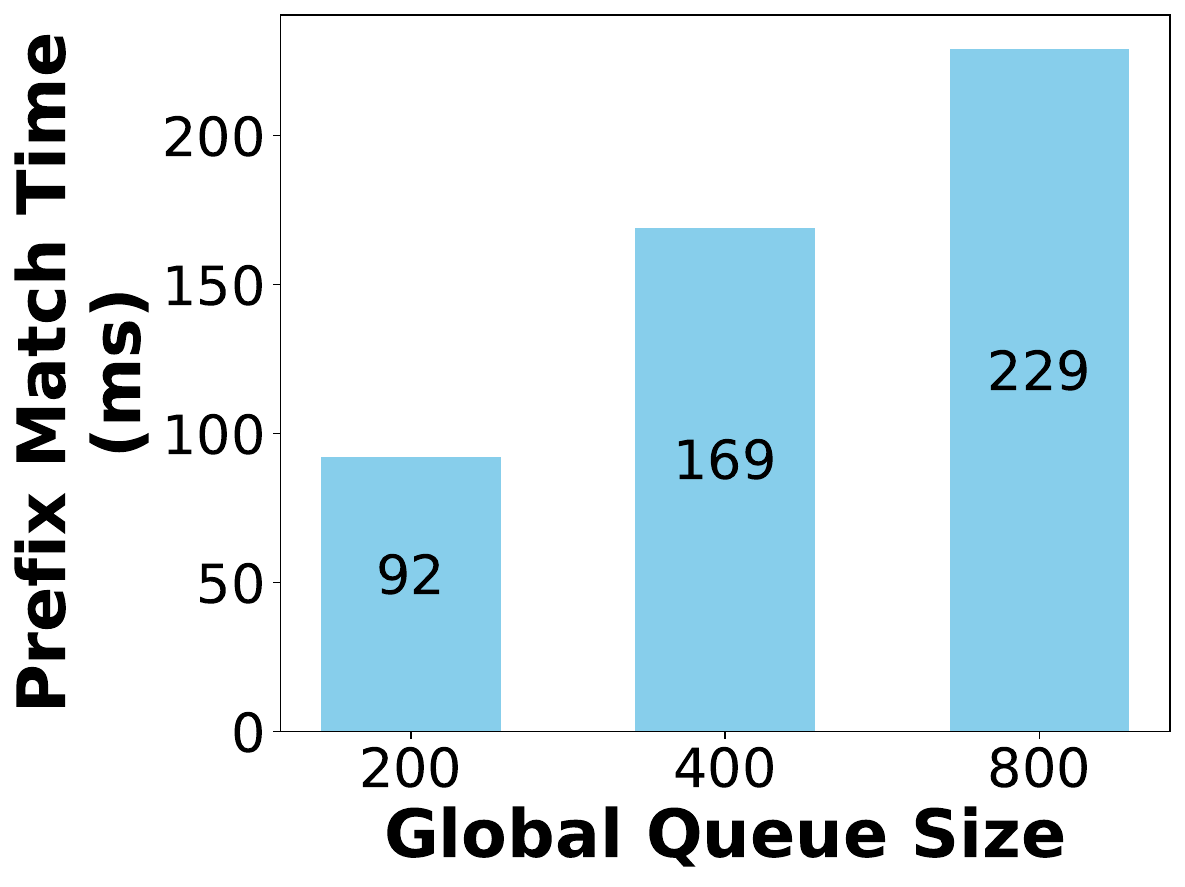}\label{fig:prefix_overhead}}
\caption{Global scheduler overhead breakdown w.r.t data parallelism degree and global queue size. The time for one decode step with bs=25 is also reported for reference. Existing serving engines such as vLLM~\cite{kwon2023efficient} and SGLang~\cite{zheng2023efficiently} normally perform a continuous batching step after multiple (e.g., 10 in SGLang) decoding steps.}
\vspace{-1em}
\label{fig:overhead}
\end{figure}

\paragraph{Overhead Analysis}
The global scheduler's overhead primarily stems from synchronization overhead, algorithmic overhead (e.g., the frequent tree-matching overhead for the global waiting queue), and metadata updates overhead. Among these, the metadata updates overhead per worker remains relatively constant as the system scales. However, the synchronization and algorithmic overhead increase dramatically as the data parallelism degree ($D$)\footnote{Here data parallelism degree refers to the number of model replicas in the distributed settings.} and global queue size increases, as shown in \cref{fig:overhead}. 
The prefix matching process (algorithmic overhead) involves matching all incoming requests in the global waiting queue against each worker's radix tree and sorting them based on prefix length to determine the dispatch order. The "Prefix Match" time (blue) increases significantly as the global queue size increases (\cref{fig:prefix_overhead}), which is normally the case when the data parallelism degree grows.

Overall, \Cref{fig:overhead} demonstrates how synchronization and algorithmic overheads dominate as the data parallelism degree increases, particularly for higher degrees ($D$ = 8) -- they add to around 40\% decoding overhead in the demonstrated case. This analysis underscores the challenges of designing scalable global schedulers to mitigate synchronization and algorithmic bottlenecks as the system scales.

Besides the significant scheduling overhead, the Global DLPM scheduler also requires extensive modification of the local worker to enable local-global information synchronization and the blocking operation to wait for the global scheduler dispatching requests.

\subsection{Our Solution: Decentralized DLPM}
\label{ssec:decen_dlpm}

\begin{figure}[!t]
\centering
\includegraphics[width=0.38\textwidth]{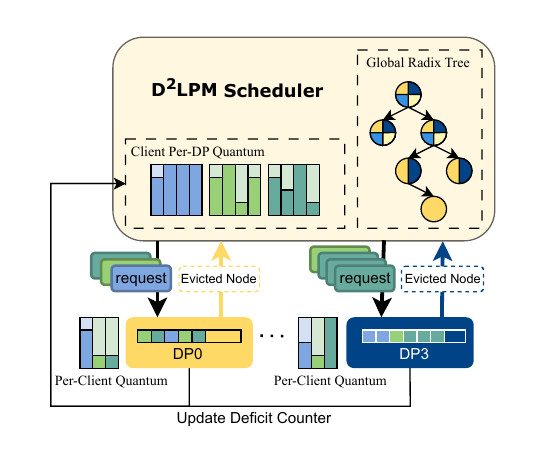}
\caption{An overview of the \distname scheduler. The global scheduler tracks the deficit counters for each client per worker to control the ``stickiness'' of a client to a worker. The local schedulers maintain the deficit counters for each client to enforce the fair sharing of the local GPU resources.}
\label{fig:doubleq}
\vspace{-1em}
\end{figure}

To mitigate \emph{the global scheduling overhead and tight coupling} between the global scheduler and the local worker, we resort to decentralized scheduling: dispatching the requests directly to local workers and queueing them at the local worker instead of the global scheduler. 
Most of the existing distributed schedulers for LLM serving (e.g., Preble~\cite{preble} and SGLang~\cite{zheng2023judging}) follow this design.

In such a decentralized design, the local worker can directly run a fair scheduling algorithm (e.g., DLPM); as long as the global scheduler can balance the per-client service on all the local workers, we could achieve global fairness guarantees~\cite{bejerano2004fairness}. Previous works in CPU scheduling~\cite{cfs} and wireless LANs bandwidth sharing~\cite{bejerano2004fairness} also demonstrate the effectiveness of such design. 
Therefore, the challenge now becomes how to strike a good trade-off between load balancing and locality. 

\begin{algorithm}[!t]
\caption{\distname Scheduling}
\begin{algorithmic}[1]
\State let $s_{w}$ denotes the current queue size of worker $w$.
\State $W \leftarrow \Call{GetWorkers}$,
$R \leftarrow \Call{InitRadixTree}{\vert W \vert}$

\Function{SelectWorker}{$G$, $i$}\label{line:select_worker}
\State $G_{avail} \leftarrow \{ w \mid q_{i,w} > 0 \}$
\While{$G_{avail} == \emptyset$}
\ForAll{$w \in W$}
$q_{i,w} \leftarrow q_{i,w} + Q^w$
\EndFor
\EndWhile
\State $G_{cand}  \leftarrow G \cap G_{avail}$
\If{$G_{cand}==\emptyset$}
\Return $\argmin_{w \in G_{avail}} s_{w}$ \label{line:routeEven}
\EndIf
\Return $\argmin_{w \in G_{cand}} s_{w}$
\EndFunction

\State $\triangleright$ \texttt{with concurrent stream 1:}
\While{True}
    \If {new request $r$ from client $i$ arrived}
        \State $G \leftarrow \Call{R.LongestMatchWorkers}{r}$
        \label{line:longest_match}
        \State $w \leftarrow$ \Call{SelectWorker}{$G$, $client(r)$}
        \State \Call{Dispatch}{$w$, $r$}
        \State $q_{i,w} \leftarrow q_{i,w} - w_e \cdot r.\text{input\_tokens}$
        \State $s_w \leftarrow s_w + 1$
    \label{line:subtract_Q_input}
        \State \Call{R.Insert}{$r$.input\_tokens, $w$}
    \EndIf
\EndWhile
\State $\triangleright$ \texttt{with concurrent stream 2:}
\While{True}
    \If{request $r$ from client $i$ has finished at worker $w$}
        \State $q_{i,w} \leftarrow q_{i,w}$ - $w_q \cdot r$.output\_tokens
        \State $s_w \leftarrow s_w - 1$
        \label{line:subtract_Q}
    \EndIf
\EndWhile
\State $\triangleright$ \texttt{with concurrent stream 3:}
\While{True}
\If{prefix $P$ has been evicted at worker $w$}
\State \Call{R.Evict}{$P$, $w$}
\label{line:evict}
\EndIf
\EndWhile
\end{algorithmic}
\label{alg:radix-tree}
\end{algorithm}

\paragraph{Double Deficit LPM (\distname)}
Our key insight is to prioritize locality first until certain limits are met: we use the quantum mechanism again to avoid a client becoming too sticky to a single worker due to the prefix cache locality by assigning quantum to each worker for each client. 
As demonstrated in \cref{fig:doubleq} and \cref{alg:radix-tree}, for each new request, the global scheduler first matches it with the global radix tree and get the workers $G$ that have its longest-matched prefix (Line~\ref{line:longest_match}). Then in the \textsc{SelectWorker} function (Line~\ref{line:select_worker}), if any wokers in $G$ has deficit counter larger than $0$ ($G_{avail}$), the worker with minimum queue size in $G\cap G_{avail}$ will be chosen. Otherwise, the worker with minimum queue size in $G_{avail}$ will be selected. After each request is dispatched, the request's input tokens will be inserted into the global radix tree and the corresponding deficit counter will be updated (Line~\ref{line:subtract_Q_input}). The global scheduler will periodically update corresponding deficit counter when there are requests finished (Line~\ref{line:subtract_Q}) as well as prune the global radix tree with collected local workers' eviction information (Line~\ref{line:evict}). Note that unlike the centralized DLPM where the eviction information needed to be passed to the global scheduler synchronously, in \distname this happens asynchronously with negligible overhead. 

We note that our \distname scheduling (with local workers running DLPM)\footnote{Generally, in \distname, the local worker can run any other fair scheduling algorithms such as VTC. In this paper, \distname specifically refers to the implementation using DLPM at the local workers.} provides global fairness guarantees corresponding to the properties introduced in \cref{sec:preliminary} through the following theorems. 
\begin{restatable}[\textbf{Service bound between backlogged clients}]{theorem}{doubleQBoundedServiceDifference} At any time interval $[t_1,t_2)$, $\max_{i} W_i(t_1, t_2) - \min_{i} W_i(t_1, t_2) \leq 2 \cdot \vert W  \vert \cdot (U + Q^u)$. The difference between the maximum service among all backlogged clients and the minimum service among all backlogged clients is bounded by $2 \cdot \vert W  \vert \cdot (U + Q^u)$, where $\vert W \vert$ is the number of workers. 
\end{restatable}

\begin{restatable}
[\textbf{Service bound between backlogged and non-backlogged clients}]{theorem}{doubleQBacklog} Consider any execution of the \distname scheme. Client $f$ that is continuously backlogged during time interval $[t_1,t_2)$ should not receive less service than another client, g, that is not continuously backlogged during the same time interval, where
$W_g(t_1,t_2) - W_f(t_1,t_2) \leq 2 \cdot (U + Q^u) \cdot \vert W \vert$.
\end{restatable}

Since there are no requests rejected to enforce fairness, \distname scheduling is work-conserving.

\section{Evaluation}
\label{sec:experiment}
\subsection{Setup}

\paragraph{Implementation}
We implement our DLPM and \distname schedulers in Python on top of SGLang~\cite{zheng2023efficiently}, a fast industry-standard LLM inference system.

\paragraph{Models and Hardware}
Our evaluation is conducted on the widely-used model \texttt{Llama-3.1-8B} and \texttt{Llama-3.2-3B}~\cite{dubey2024llama}. Other transformer-based LLMs such as Qwen~\cite{yang2024qwen2}, DeepSeek~\cite{deepseek}, and Mistral~\cite{jiang2023mistral} share a similar backbone architecture and are also compatible with our system.
For hardware, we test on NVIDIA A100 80GB and A10G GPUs.

\begin{table}[ht]
\centering
\footnotesize
\caption{Workload configurations.}
\vspace{-1em}
\resizebox{\linewidth}{!}{
\setlength{\tabcolsep}{2pt} 
\begin{tabular}{ccccc}
\toprule
\textbf{Workload} & \textbf{Dataset} & \textbf{Avg Prefix Len.}  & \textbf{Avg Output Len.} \\
\midrule
Long-Context QA & LooGLE~\cite{li2023loogle} & 21449  & 15 \\
Tree-of-Thoughts & GSM8K~\cite{cobbe2021training} & 546
 & 256 \\
LLM-as-a-Judge & Synthetic articles~\cite{zheng2023efficiently} &  2701

  & 256 \\
Real Multi-Turn & Chatbot Arena~\cite{zheng2023judging} & 56 & 142 \\
\bottomrule
\end{tabular}
}
\label{tab:workload_setting}
\end{table}

\paragraph{Workloads and Datasets}
We evaluate the efficiency and effectiveness of the schedulers on 4 diverse LLM-based workloads, each characterized by its unique execution graph structures (\cref{fig:demo}) and variations in prefix and output length distributions. as detailed in \cref{tab:workload_setting}. Specifically, we evaluate long document understanding using the LooGLE~\cite{li2023loogle} dataset. We implement the Tree-of-Thoughts~\cite{yao2024tree} program for solving GSM8K~\cite{cobbe2021training} problems (with a tree height of 4), and the LLM-as-a-Judge~\cite{zheng2023efficiently} program, which utilizes the branch-solve-merge technique to evaluate synthetic articles. We also conduct experiments on real-world multi-turn conversation traces from Chatbot Arena~\cite{zheng2023judging}.

\begin{figure}[ht]
\centering
\includegraphics[width=0.47\textwidth]
{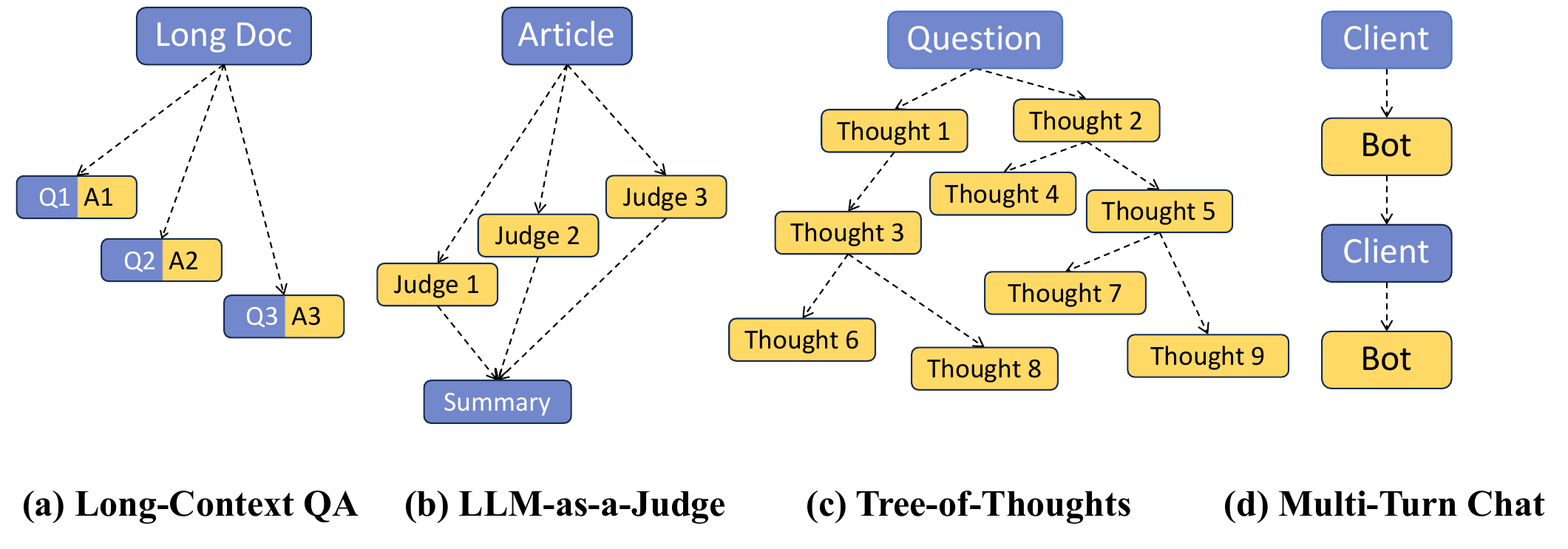}
\caption{Illustration of the execution graphs of different workloads in our benchmark. }
\vspace{-1em}
\label{fig:demo}
\end{figure}

\paragraph{Synthetic Traces} For Long-Context QA, Tree-of-Thoughts, and LLM-as-a-Judge, we generate synthetic client request traces following the Gamma process, as done in~\cite{vtc,sheng2023slora,li2023alpaserve}, with the request rate increasing as the number of GPUs scales.

For these three workloads, we evaluate two distinct types of misbehaving patterns, as detailed in \cref{tab:model_setting}. The first type (S1) involves a misbehaving client sending more requests than well-behaved clients. Specifically, although all clients send programs at the same request rate, the misbehaving client submits programs with a more complex execution graph (\eg, more branches in Tree-of-Thoughts). The second type (S2) features a misbehaving client sending programs with the same structural complexity and at the same request rate as well-behaved clients, but with the input altered to increase the prefix length. These workloads are evaluated with the \texttt{Llama-3.1-8B} model served on NVIDIA A100 GPUs. The related results are reported in \cref{sec:synthetic}.

\begin{table}[ht]
\centering
\footnotesize
\caption{Synthetic workload configurations. \faThumbsODown\ stands for misbehaving client and \faThumbsOUp\ denotes well-behaved clients. 
}
\vspace{-1em}
\resizebox{0.98\linewidth}{!}{
\begin{tabular}{cl}
\toprule
\textbf{Workload}   &\textbf{Detailed Behavior} \\
\midrule
\rowcolor{black!10} \multicolumn{2}{l}{\textit{S1: More Requests}} \\
Long-Context QA  & \faThumbsODown: Higher req rate\\
\multirow{2}{*}{Tree-of-Thoughts}  & \faThumbsODown: Trees of 4 branches (340 req per tree) \\
                                  & \faThumbsOUp: Trees of 2 branches (30 req per tree)\\
\multirow{2}{*}{LLM-as-Judge}     & \faThumbsODown: Evaluation with 16 dimensions\\
                                  & \faThumbsOUp: Evaluation with 2 dimensions\\
\rowcolor{black!10} \multicolumn{2}{l}{\textit{S2: Longer Prefix}} \\                             
Long-Context QA  & \faThumbsODown: $2\times$ longer input documents\\
Tree-of-Thoughts  & \faThumbsODown: $10\times$ longer input questions \\
LLM-as-Judge     & \faThumbsODown: Extra 600 tokens before each article\\
\bottomrule
\end{tabular}
}
\label{tab:model_setting}
\end{table}

\paragraph{Real-world Traces} For real-world multi-turn conversation, we re-scale the request time stamps provided in the dataset\footnote{\url{https://huggingface.co/datasets/lmsys/chatbot_arena_conversations}} and aggregate multiple clients’ requests to closely mimic high-demand scenarios. This workload is evaluated with the \texttt{Llama-3.2-3B} model served on NVIDIA A10G GPUs.
The related results are reported in \cref{sec:real}.

\begin{figure*}[!t]
	\centering
	\includegraphics[width=\textwidth]{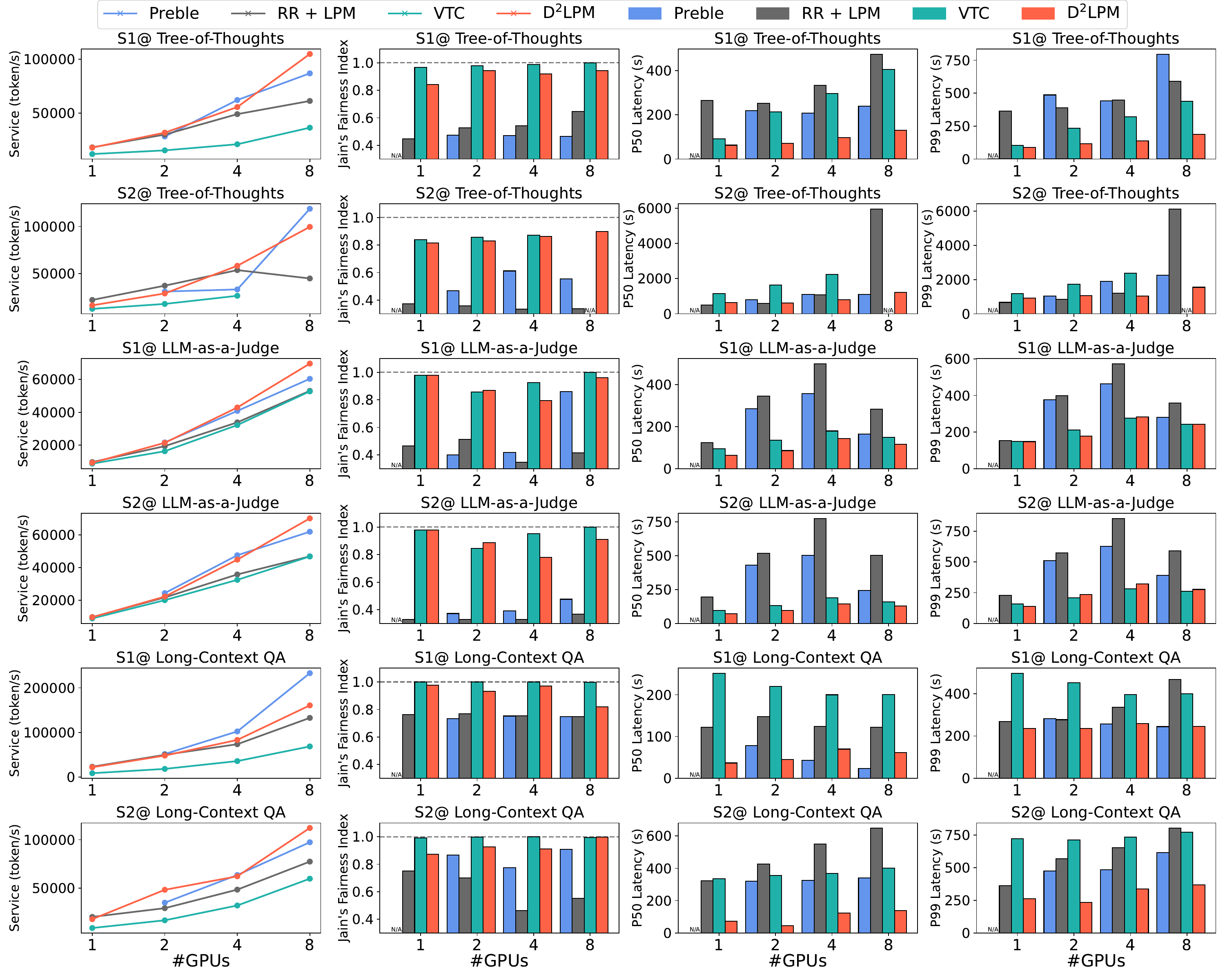}
	\vspace{-1mm}
\caption{Summary of results across three datasets and two types of misbehaving clients on up to 8 A100 GPUs (8B model). The reported latency represents the average latency for well-behaved clients. The data point for S2@Tree-of-Thoughts with $D = 8$ is omitted, as it takes too long to complete.}
	\label{fig:main plot}
	\vspace{-2mm}
\end{figure*}

\paragraph{Baselines}
We compare DLPM and \distname with three baseline scheduling algorithms. 
\begin{itemize}
    \item \textbf{\distname}: The local worker adopts DLPM, and the global scheduler runs the \distname algorithm when Data Parallelism Degree $D > 1$.
    \item \textbf{RR + LPM}: The local scheduler runs LPM, and the global scheduler uses the round-robin (RR) algorithm when \(D>1\). It is the default distributed scheduling algorithm in SGLang~\cite{zheng2023efficiently} without fairness guarantees.
    \item \textbf{Preble}~\cite{preble}: Preble is a state-of-the-art distributed LLM serving system that aims to provide high serving throughput by balancing load distribution and locality, yet without fairness guarantee. Specifically, it dispatches requests based on a pre-defined prefix-matching ratio to decide whether to explore a new GPU or exploit locality.
    \item \textbf{VTC}~\cite{vtc}: The local scheduler runs VTC, and the global scheduler applies a per-client round-robin strategy when \(D>1\). 
    Extending VTC with round-robin scheduling is the straightforward approach to ensuring fairness in distributed settings, with fairness bound proven in \cref{sec:DRR client}. 
\end{itemize}

\paragraph{Metrics} To measure the system efficiency and fairness achieved by different scheduling algorithms, we use the following three metrics:
\begin{itemize}
    \item \textbf{Service Rate}: We measure the clients' service as a weighted sum of the number of input tokens and the number of output tokens, following VTC~\cite{vtc}\footnote{Note that this service rate is from clients' perspective. From the system's perspective, the actual service is measured by the cost function using the number of extend tokens.}. As discussed in \cref{sec:preliminary}, the weight for input token is $1$ and the weight for output token is $2$.
    \item \textbf{Jain's Fairness Index}~\cite{jain1984quantitative} is a widely-used metric for evaluating the fairness of resource allocation in networked systems~\cite{fairnessinnetwork}. The index is mathematically defined as:
    \[
    J(x_1, x_2, \ldots, x_n) = \frac{\left(\sum_{i=1}^n x_i\right)^2}{n \sum_{i=1}^n x_i^2},
    \]
    where $x_i$ represents the allocation for the $i^{\text{th}}$ client, and $n$ is the total number of clients. The value of $J$ ranges from $\frac{1}{n}$ (minimum fairness, when one client monopolizes all resources) to 1 (maximum fairness, when resources are equally distributed). 
    In our context, we compute the Jain's Fairness Index by letting $x_i$ denote the service rate of client $i$. The calculation is based on the time interval during which all clients are active, ensuring an accurate representation of fairness across the system. 
    \item \textbf{P50 and P99 Latency}: We assess the scheduler's effectiveness in maintaining service quality for well-behaved clients by measuring their P50 and P99 latency.
    We measure latency using the end-to-end completion time for program evaluation. 
    We use the TTFT (Time to First Token) latencyc\footnote{For QA tasks, a shorter TTFT contributes to improved client experience.} metric for long-context QA tasks. 
\end{itemize}

\subsection{Results on Synthetic Traces}\label{sec:synthetic}

We present all three metrics across three workloads and two types of misbehaving clients in \cref{fig:main plot}. Both VTC and \distname provide theoretical fairness guarantees, whereas Preble and RR + LPM do not. The data point for Preble with \(D=1\) is omitted because Preble is designed as a multi-GPU cache-aware prompt dispatch system.

\paragraph{Throughput Analysis}
As previously discussed, ensuring fairness inherently competes with maximizing throughput. However, \distname's effective global cache-aware scheduler and DLPM enable significant performance gains, achieving up to a 2.87$\times$ improvement compared to the only other fair algorithm, VTC.

\distname achieves better throughput than RR + LPM, with improvements of up to 2.22$\times$. In the case of S2@Tree-of-Thoughts, the poor performance of RR + LPM with $D = 8$ compared to $D = 4$ can be attributed to the complex sharing patterns inherent in Tree-of-Thoughts. Round Robin fails to preserve locality among GPUs, leading to a significant drop in cache hit rate (\ie, from $95\%$ to $50\%$). This limitation indicates that RR + LPM does not scale effectively when clients submit complex LLM programs.

Compared to Preble, \distname consistently matches or exceeds its performance across all workloads and GPU configurations, demonstrating its ability to sustain high throughput while ensuring fairness. An exception arises for S2@Long-Context QA with $D = 8$, where Preble outperforms \distname in throughput. This discrepancy occurs because \distname sacrifices some locality to maintain fairness, resulting in increased prefix recompute overhead. As indicated in \cref{tab:workload_setting}, the LooGLE dataset features an exceptionally high prefix length-to-output length ratio. In this case, the cost of recomputing long documents becomes substantial, with the prefill stage significantly dominating the generation time. Consequently, the Long-Context QA workload serves as a worst-case scenario that adversely impacts \distname's throughput. However, when the prefix length-to-output length ratio falls within a reasonable range, the \distname algorithm consistently matches and even slightly surpasses the performance of state-of-the-art non-fair scheduling algorithms. This is achieved through the careful management of load balance and locality trade-offs in the global scheduler, as well as locality and fairness trade-offs in the local scheduler.

\paragraph{Jain's Fairness Index Analysis}
From the second column in \cref{fig:main plot}, it is clear that \distname consistently outperforms both Preble and RR + LPM. This is because \distname provides strict fairness guarantees. However, it is slightly less fair than VTC, as \distname relaxes the fairness bounds to improve locality, which leads to higher throughput but slightly worse fairness control. 
Preble performs slightly better than RR + LPM due to its multi-level priority wait queue, which avoids starvation but cannot provide isolation and strict fairness guarantees. As a result, there remains a notable gap between Preble and \distname.

\paragraph{Well-behaved Clients' Latency Analysis}
We use the average P50 and P99 latency of well-behaved clients to evaluate the experience of well-behaved clients when a misbehaving client is present. Algorithms focusing on high system efficiency might inadvertently increase latency for well-behaved clients as these schedulers may prioritize the requests from the misbehaving clients to optimize the prefix cache hit rate.
Preble and RR + LPM, therefore, can result in up to $7.18\times$ and $9.55\times$ higher latency, respectively, compared to \distname. On average, \distname achieves $2.90\times$ and $4.06\times$ lower latency than Preble and RR + LPM.
On the other hand, algorithms that focus solely on fairness will also incur high latency for well-behaved clients due to reduced overall system efficiency.
For instance, VTC can lead to latency up to $7.96\times$ higher than \distname, with an average latency increase of $2.98\times$.

\subsection{Results on Real-world Traces}\label{sec:real}
\begin{figure}[!t]
    \centering
    \includegraphics[width=\columnwidth]{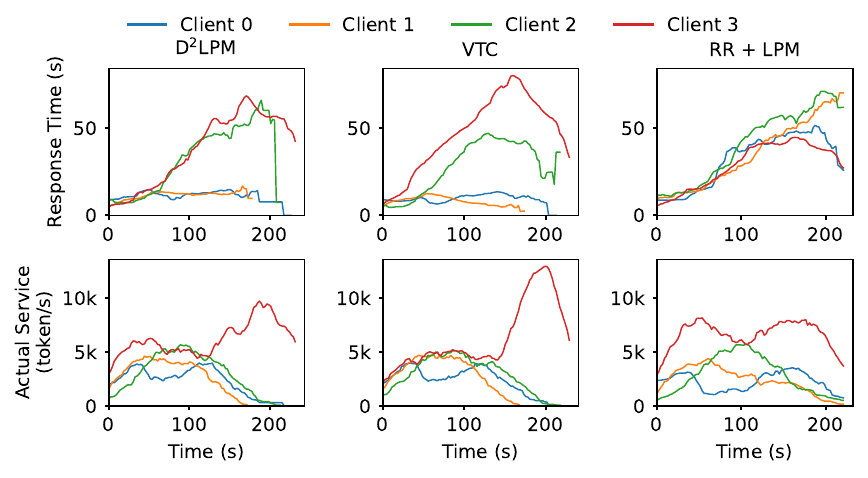}
    \caption{Fairness and performance visualization for the real-world multi-turn conversation workload ($D=2$). Clients 2 and 3 send requests at a much higher rate than Clients 0 and 1.
    }
    \label{fig:multi-turn-exp}
\end{figure}

Figure~\ref{fig:multi-turn-exp} demonstrates the fairness and performance comparison of different schedulers on the real-world multi-turn conversation workload. 
In this workload, Clients 2 and 3 initially send excessive number of requests, and Client 2 returns to normal midway. A fair scheduler should prevent these two clients from impacting other clients. 
Due to LPM’s prioritization strategy, which favors requests with longer prefix matches, Clients 2 and 3 receive a disproportionately large share of resources. As a result, Clients 0 and 1 suffer from high response times and reduced throughput. In contrast, VTC achieves relatively low response times and maintains high throughput for Clients 0 and 1. However, such strict fair allocation comes at the expense of Clients 2 and 3, who endure substantial response delays, reaching up to 80 seconds.

\distname achieves a more reasonable distribution of resources, protecting well-behaved clients from the disruptive effects of high request rates by the misbehaving clients. \distname ensures consistently low response time and high throughput for both well-behaved and previously misbehaving (\ie, Client 2) clients. Thus, \distname not only mitigates the impact of malicious usage patterns but also improves overall system performance and fairness compared to the baseline approaches.

\section{Ablation Studies}

\subsection{Visualization of Fairness Properties}\label{sec:tot_vis}

We visualize the response time and the services provided by the server to different clients over time in \cref{fig:tot_visual}. The experiments use 4 A10G GPUs as the testbed, with all clients sending Tree-of-Thoughts programs at the same rate and with a consistent branch count of 3. However, client 0 is misbehaving by sending a longer prefix, i.e. $10\times$ longer than well-behaved clients. The maximum value on the x-axis represents the end-to-end completion time of all programs. As observed, \distname achieves the shortest execution time, demonstrating up to $2\times$ speedup compared to VTC and Preble.

\begin{figure}[ht]
    \centering
    \includegraphics[width=\columnwidth]{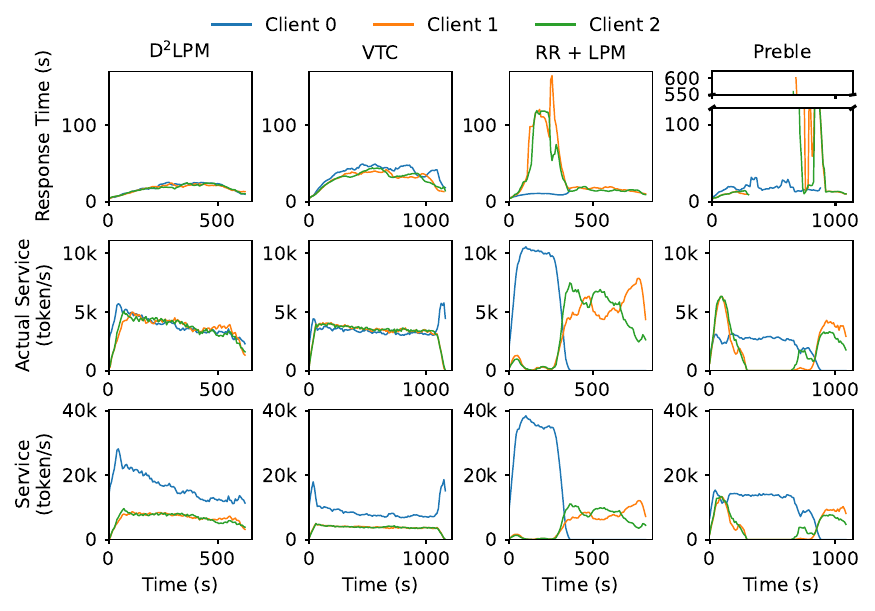}
    \caption{Fairness and performance visualization of different schedulers on Tree-of-Thoughts workloads with $D = 4$ (3B model + 4 A10G GPUs). The maximum value on the X-axis represents the end-to-end completion time for each scheduler. The actual service is calculated using the cost function defined in \cref{sec:preliminary}, which is a weighted sum of the number of extend tokens and the number of output tokens.}
    \label{fig:tot_visual}
\end{figure}

From the first row of the figure, we observe that \distname consistently maintains lower response times compared to VTC by preserving a higher degree of locality, enhancing overall system efficiency.
Furthermore, it avoids the excessively high response time caused by schedulers like RR + LPM and Preble, which lack fairness control. These schedulers tend to prioritize serving client 0, resulting in substantial delays for other clients. For instance, as shown in the figure, the service received by clients 1 and 2 between 100s and 700s is almost zero for Preble, causing a queuing latency of up to 600 seconds.

The second and third rows depict the actual service and the service received by each client, respectively. As shown in the second row of \cref{fig:tot_visual}, both VTC and \distname achieve an ideal sharing of resources across the 4 GPUs in terms of actual service. In the third row, we can observe that the service rate of client 0 is higher than clients 1 and 2 -- this is because client 0 has longer prefix sharing and thus lower cost per token. However, due to the relatively low cache hit rate of VTC, it experiences worse end-to-end performance. In contrast, the other two algorithms demonstrate significant unfairness in resource allocation across clients.

A key highlight here is the extremely low throughput observed with Preble. Preble prioritizes dispatching requests to the GPU with the longest prefix-matching length, provided the matching length exceeds a predefined threshold. Between 300 and 600 seconds, client 0's requests are continuously dispatched to a single GPU as the prefix-matching ratio will always exceed the pre-defined threshold. Some requests from clients 1 and 2 get queued at this monopolized GPU, which blocks these clients from generating new requests (i.e., "deeper" thoughts), due to the inherent LLM call dependencies in the Tree-of-Thoughts programs. This results in severe workload imbalance among the GPUs, with the cluster at merely 1/4 of its potential computational capacity.

\subsection{Impact of $Q^{u}$ in DLPM and $Q^{w}$ in \distname}
\begin{figure}[ht]
    \centering
    \begin{subfigure}[t]{0.50\columnwidth}
        \centering
        \includegraphics[width=\linewidth]{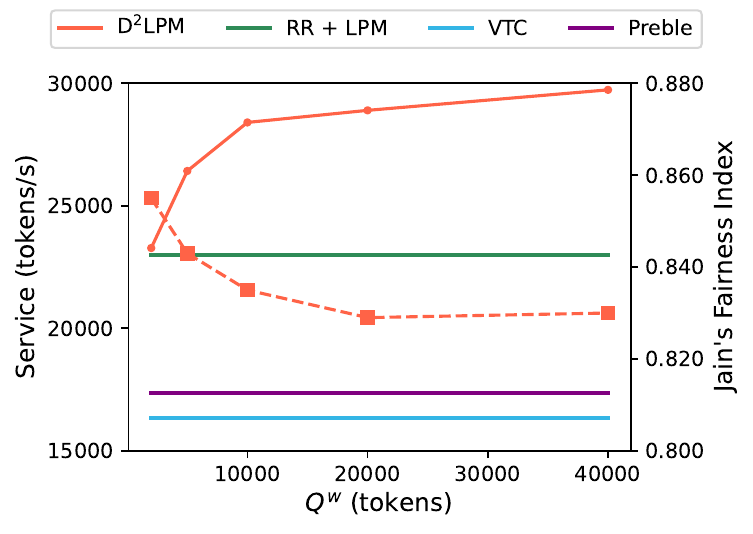}
        \caption{Throughput of Tree-of-Thoughts with one misbehaving client.}
        \label{fig:subfig1}
    \end{subfigure}
    \hfill
    \begin{subfigure}[t]{0.48\columnwidth}
        \centering
        \includegraphics[width=\linewidth]{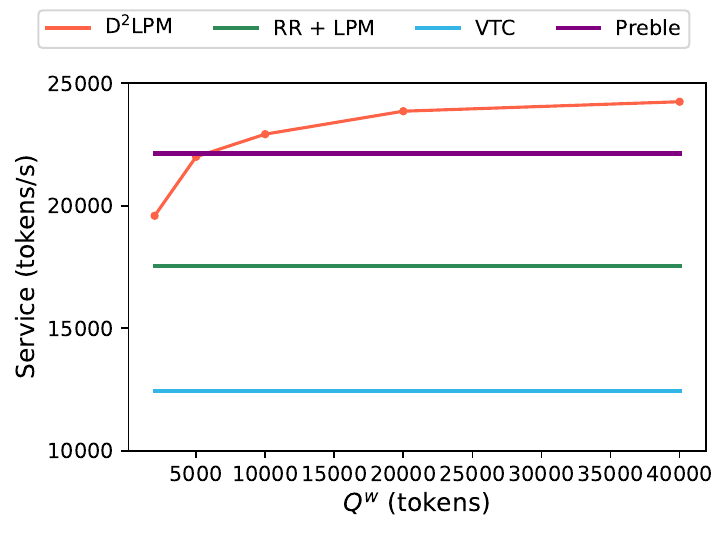}
        \caption{Throughput of Tree-of-Thoughts with all well-behaved clients.}
        \label{fig:subfig2}
    \end{subfigure}
\caption{Impact of $Q^{w}$ on throughput under different workloads ($D=4$). The solid line represents throughput, while the dashed line represents Jain's Index. The fairness index in (b) is omitted as it consistently equals 1. }

    \label{fig:globalq}
\end{figure}

We now examine the trade-off between locality and fairness using $Q^{u}$ and $Q^{w}$. The impact of $Q^{u}$ is illustrated in \cref{fig:pareto}, where increasing $Q^{u}$ enhances throughput but compromises fairness control. By adjusting the value of $Q^{u}$, the server can achieve a tailored trade-off between performance and fairness, defining a new Pareto frontier compared to VTC and LPM.

\cref{fig:globalq} illustrates the impact of $Q^{w}$ on throughput in \distname. To recap, $Q^{w}$ represents the quantum of service assigned to each worker in \distname, where a larger $Q^{w}$ typically implies a better locality for requests within a client. As shown in \cref{fig:globalq}, as $Q^{w}$ increases, the throughput of \distname also increases, eventually stabilizing and surpassing all other schedulers. The low throughput of Preble, as seen in \cref{fig:subfig1}, has been explained earlier in \cref{sec:tot_vis}.

Although $Q^{w}$ is not in the fairness bound of \distname as demonstrated in \cref{sec:Doubleq bound}, it does slightly affect Jain's Fairness Index. Specifically, the index decreases from 0.855 to 0.83 when $Q^{w}$ increases from 2000 to 40000, due to the more unbalanced dispatching of requests within a client \footnote{When $Q^{w}$ is set to infinity, the algorithm is reduced to be similar as Preble, which lacks fairness guarantees since the difference in load across workers becomes unbounded, as proven in Theorem~\ref{theorem:infiniteQ}}.

\subsection{Scaling the Number of Clients}

We assess DLPM's performance as we increase the number of clients from 5 to 50, using a single A10 GPU as the testbed, while maintaining a constant total request rate. As depicted in \cref{fig:scaleuser}, DLPM consistently achieves a service rate comparable to LPM, even as the number of clients increases, whereas VTC consistently underperforms.

\begin{figure}[ht]
    \centering
    \includegraphics[width=0.75\columnwidth]{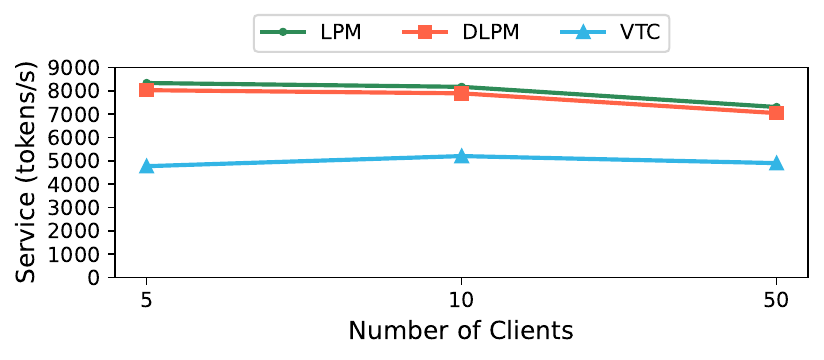}
    \caption{Service rate w.r.t the number of clients on a single A10 GPU (3B model).}
    \label{fig:scaleuser}
\end{figure}

Note that as the number of clients rises, the number of distinct prefixes in the same volume of requests increases, which marginally reduces the cache hit rate for both DLPM and LPM, leading to a slight decrease in service rate as the number of clients increases.
In contrast, VTC's performance is less affected since its cache hit rate is consistently low regardless of the number of clients.

\subsection{Mix of Workloads}
As a complement to the single-workload scenario discussed earlier, we now explore a more realistic setting where clients handle diverse workloads. As shown in \cref{fig:mix}, DLPM consistently achieves better response time and end-to-end execution times compared to the other schedulers. In the LPM scheduler, clients sending Tree-of-Thoughts programs act as misbehaving clients, significantly increasing the response time for other clients. From the second row, we observe that VTC exhibits better fairness control than DLPM, as it provides more evenly distributed actual service across clients. This demonstrates that DLPM sacrifices some degree of fairness to achieve higher throughput.

\begin{figure}[ht]
    \centering
    \includegraphics[width=\columnwidth]{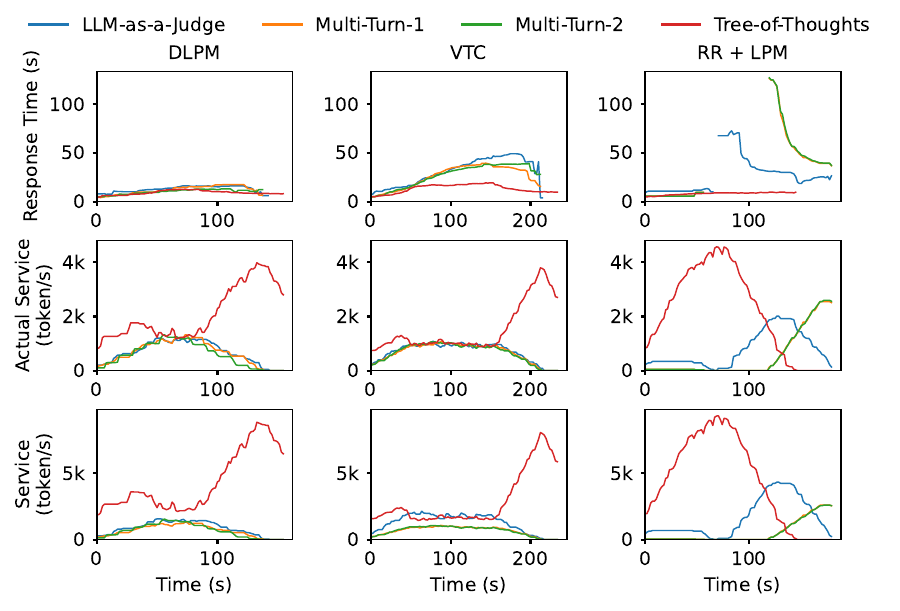}
    \caption{Mix of workloads among four clients: two engage in multi-turn conversations, while the other two send different programs, all within a single-GPU setup (3B model + an A10G GPU).}
    \label{fig:mix}
    \vspace{-1em}
\end{figure}

\section{Related Work}
\paragraph{Fairness in ML Workloads}
ML training workloads have extensively studied the fairness problems in shared clusters~\cite{mahajan2020themis,narayanan2020heterogeneity,chaudhary2020balancing,qiao2021pollux}. 
Due to their unique characteristics such as long running time, placement sensitivity, and statistical efficiency (\ie, the amount of progress per unit of data consumed), traditional fair scheduling for big data workloads~\cite{isard2009quincy,robert2016carbyne} does not work well. 
To handle the long-running and placement-sensitive natures of ML training workloads, Themsis~\cite{mahajan2020themis} proposes new finish-time fairness metrics, and leverages multi-round partial allocation auctions to provide Pareto-efficient and envy-free resource allocations. 
To consider statistical efficiency for higher cluster-wide resource utilization, Pollux~\cite{qiao2021pollux} introduces goodput-driven cluster scheduling by jointly optimizing resource allocations and job batch sizes. 
On the other hand, prior work VTC and our work focus on the LLM inference-time fairness. 
Compared to VTC, our work co-optimizes both fairness and prefix sharing for higher performance without losing fairness. 

\paragraph{Fairness in Other Workloads}
Fairness is a long-existing topic in networking and operating systems. 
For example, networking needs to guarantee fairness among different switching ports~\cite{drr} and during link bandwidth allocation~\cite{wfq,pgps,sfq,sfq_d,self-clock}; OS scheduling needs to guarantee fair CPU time share among different processes~\cite{linux_cfs, eevdf}, and fair memory allocations~\cite{nesbit2006fairmemory}. 
Fairness is also extensively studied in big data workload scheduling with prominent prior work of Delay Scheduling~\cite{delayschedule} and Dominant Resource Fairness~\cite{drf}. 
Our fair scheduling design is inspired by many prior work such as Deficit Round Robin~\cite{drr} and Delay Scheduling~\cite{delayschedule}; but differently, we explicitly optimize for the prefix sharing property in LLM inference workloads while guaranteeing fairness. 

\paragraph{Locality in LLM Inference} 
Previous advances in LLM inference focus on batching and memory optimization~\cite{kwon2023efficient, yu2022orca}. SGLang further exploits locality in scheduling to improve LLM inference performance for emerging applications such as multi-turn chatting~\cite{zheng2023efficiently}. 
It leverages the LPM scheduling with RadixTree to save GPU memory and avoid redundant computations through prefix sharing. 
Preble~\cite{preble} further extends LPM into distributed settings to jointly optimize load balancing and prefix caching locality for high throughput. 
BlendServe~\cite{blendserve} co-optimizes GPU resource overlapping and prefix sharing for offline LLM inference, achieving nearly optimal inference throughput. 
Unlike the above work which only focuses on inference throughput, our work presents a principled way of navigating the trade-off between performance and fairness in multi-client scenarios. 
\section{Limitation and Future Work}

\paragraph{Prefix Sharing Among Clients}
This work explores the prefix sharing among the inference requests of each individual client, while ignoring the prefix sharing among requests from different clients. 
Exploring how to fairly share prefix cache among different clients would be interesting future work, \eg, deciding which client should pay the quantum deduction when several clients share the same prefix. 

\paragraph{Fairness with In-Program Data Dependencies}
Our work can be further improved by considering the data dependencies between different LLM inference requests in programs. 
For example, an upstream inference request may generate the input data for many downstream requests, leading to higher parallelisms potentially with higher performance. 
In this case, it would be better to first schedule this upstream request, even if doing so may break prefix sharing or strict fairness. 
Similar data dependencies have been explored in the context of cluster scheduling for big data analytics~\cite{chung2020unearthing}.

\section{Conclusion} 
This paper introduces the first prefix-aware fair scheduling algorithm for LLM serving, namely, DLPM. 
We also propose an extension of the algorithm, \distname, to preserve locality with global fairness guarantees in a distributed setup. 
Our algorithm achieves up to 2.87$\times$ higher throughput than state-of-the-art fair scheduling algorithms in LLM like VTC, and 7.18$\times$ lower latency for victim clients compared to locality-aware scheduling algorithms like Preble.

\newpage
\bibliographystyle{plain}
\bibliography{reference}

\newpage
\appendix
\section{Appendix}

\subsection{Proof for Local DLPM}\label{sec:appendix_dlpm}

\begin{restatable}[\textbf{Service Bound}]{theorem}{dlpmServiceBound} 
\label{theorem:service-bound}
Consider any execution of the DLPM scheme in which client $i$ is backlogged. After any $K_i$ rounds (where $q_i$ is replenished $K_i$ times) from $t_1$ to $t_2$, the difference between $K_i \cdot Q^u$ (i.e., the service that client $i$ should have sent) and $W_i (t_1, t_2)$ (i.e., the service that client $i$ actually received) is bounded by $\max(Q^u, U)$, where $U=w_e\cdot L_{input} + w_q\cdot M$.
\label{theorem:dlpmServiceBound}
\end{restatable}

\begin{proof}
Let $q_i(t)$ denote the deficit counter value of client $i$ at time $t$. Since the deficit counter will only be refilled when $q_i \leq 0$ (line \ref{line:replenish_c}) by $Q^u$, we have 
\begin{equation}
q_i(t) \leq Q^u \label{eq:upper_bound_dplm}
\end{equation}

Now we prove through induction:
\begin{equation}
q_i(t) > -U \label{eq:lower_bound_dplm}
\end{equation}
\begin{itemize}
    \item At the beginning, all $q_i(0) = 0$. \Cref{eq:lower_bound_dplm} holds.
    \item We then prove if at time $t$, \Cref{eq:lower_bound_dplm} holds, then for $t' > t$, \Cref{eq:lower_bound_dplm} also holds. 
    \item At line \ref{line:replenish_c}, $q_i(t')=q_i + Q^u > q_i > -U$. \Cref{eq:lower_bound_dplm} holds.
    \item Since line \ref{line:subtract_c} will be reached only when $q_i > 0$, $q_i(t') = q_i - w_e\cdot extend\_length(r) > - w_e \cdot L_{input}$. \Cref{eq:lower_bound_dplm} holds.
    \item At line \ref{line:subtract_c_finish}, since $q_i(t') = q_i - w_q \cdot \vert \{r| client(r)=i, r \in B\} \vert$ will be repeated for $n$ steps until some requests are finished. Therefore, we have  $q_i(t') \geq q_i - n \cdot w_q \cdot \vert \{r| client(r)=i, r \in B\}\vert$. Since the number of decoded tokens cannot exceed the server's maximum token capacity $M$, $n \cdot \vert \{r| client(r)=i, r \in B\}\vert \leq M$. We then have $q_i(t') = q_i - w_q \cdot M > - U$. \Cref{eq:lower_bound_dplm} holds.
\end{itemize}

Therefore, we have $W_i(t_1, t_2)=K_i \cdot Q^u - q_i(t_2)$. Combining \Cref{eq:lower_bound_dplm} and \Cref{eq:upper_bound_dplm}, we have:
\begin{equation}
    \vert W_i(t_1, t_2) - K_i \cdot Q^u \vert = \vert q_i(t_2) \vert \leq \max(Q^u, U)
\end{equation}
\end{proof}

\begin{restatable}[\textbf{Latency Bound}]{theorem}{latencyBound} Let $A(r)$ and $D(r)$ denote the arrival time and dispatch time of a request $r$. Assume there are in total $n$ clients, $\forall t_1, t_2$, if at $t_1$, a client $f$ is not backlogged and has no requests in the running batch, then the next request $r_f$ with $t_1  < A(r_f) < t_2$ will have its response time bounded: $D(r_f) - A(r_f) \leq 2 \cdot (n-1) \cdot \frac{Q^u+U}{a}$, where $a$ is the lower bound of the server capacity. 
\label{theorem:latency-bound}
\end{restatable}

\begin{proof}
    \begin{itemize}
        \item Since there is no running batch of $f$ in the system, $r_f$ will be selected for the next request for $f$. 
        \item Earlier, we have shown that the service bound for backlogged clients compared to either backlogged or non-backlogged clients is $2 \cdot (Q^u + U)$. 
        \item From $t_1$ to $D(r_f)$, $W_f (t_1, D(r_f) )$ will be within $2 \cdot (Q^u + U)$ of service received by other clients. 
        \item Since at Line~\ref{line:client_join}, $q_f$ is set to 0 when $f$ rejoins, the maximum number of tokens served before $f$ is served again is: $2 \cdot (n-1) \cdot (Q^u + U)$, where $n - 1$ is the $n - 1$ other clients.  
        \item Given that $a$ is the lower bound of the server capacity, the dispatch time for $f$ is therefore bounded: $D(r_f) - A(r_f) \leq 2 \cdot (n-1) \cdot \frac{Q^u + U}{a}$.
    \end{itemize}
\end{proof}

\subsection{Proof for \distname Scheduling}
\label{section:doubleq}
\begin{restatable}[\textbf{Service Bound}]{theorem}{doubleqServiceBound}
Consider any execution of the \distname Scheduling scheme in which client $i$ is backlogged. The difference between $\sum_{w \in W} k_{i,w} \cdot Q^u$ (i.e., the service that client $i$ should have sent) and $W_i$ (i.e., the service that client $i$ actually received) is bounded by $\max(Q^u, U) \times \vert W \vert$, where $U=w_e\cdot L_{input} + w_q\cdot M$. Let $k_{i,w}$ is the number of times client $i$ has been replenished at worker $w$. 
\end{restatable}


\begin{proof}
Let $k_{i,w}$ be the number of times the client $i$ has replenished quantum locally at worker $w$. We want to show for a client $i$: 
$$
\sum_{w \in W} k_{i,w} \cdot Q^u - \sum_{w \in W} (k_{i,w} \cdot Q^u - q_{i,w}^u(t)) \leq \max (Q^u, U) \times \vert W \vert
$$

Let $q_{i,w}^u(t)$ denote the deficit counter value for worker $w$ of client $i$ at time $t$. Since the deficit counter will only be refilled when $q_{i,w}^u(t) \leq 0$ (line \ref{line:replenish_c}) by $Q^u$, we have 
\begin{equation}
q_{i,w}^u(t) \leq Q^u \label{eq:upper_bound_doubleq}
\end{equation}

Now we prove through induction:
\begin{equation}
q_{i,w}^u(t) > -U \label{eq:lower_bound_doubleq}
\end{equation}
\begin{itemize}
    \item At the beginning, all $q_{i,w}^u(t)=0$. \Cref{eq:lower_bound_doubleq} holds.
    \item We then prove if at time $t$, \Cref{eq:lower_bound_doubleq} holds, then for $t' > t$, \Cref{eq:lower_bound_doubleq} also holds. 
    \item At line \ref{line:replenish_c}, $q_{i,w}^u(t')=q_{i,w}^u(t) + Q^u > q_{i,w}^u(t) > -U$. \Cref{eq:lower_bound_doubleq} holds.
    \item Since line \ref{line:subtract_c} will be reached only when $q_{i,w}^u(t) > 0$, $q_{i,w}^u(t') = q_{i,w}^u(t) - w_e \cdot L_{input} > - w_e \cdot L_{input}$. \Cref{eq:lower_bound_doubleq} holds.
    \item At line \ref{line:subtract_Q}, since $q_{i,w}^u(t') = q_{i,w}^u(t) - w_q \cdot \vert \{r| client(r)=i, r \in B\} \vert$ will be repeated for $n$ steps until some requests are finished. Therefore, we have  $q_{i,w}^u(t') = q_{i,w}^u(t) - n \cdot w_q \cdot \vert \{r| client(r)=i, r \in B\}\vert$. Since the number of decoded tokens cannot exceed the server's maximum token capacity $M$, $n \cdot \vert \{r| client(r)=i, r \in B\}\vert \leq M$. We then have $q_{i
    ,w}^u = q_{i,w}^u(t) - w_q \cdot M > - w_e \cdot L_{input} - w_q \cdot M$. \Cref{eq:lower_bound_doubleq} holds.
\end{itemize}

We have $W_i(t_1, t_2) =\sum_{w \in W} (k_{i,w} \cdot Q^u - q_{i,w}^u(t_2))$. Combining \Cref{eq:lower_bound_doubleq} and \Cref{eq:upper_bound_doubleq}, we have:
\begin{equation}
    \vert \sum_{w \in W} k_{i,w} \cdot Q^u - W_i(t_1, t_2)  \vert \leq \max (Q^u, U) \times \vert W \vert
\end{equation}
\end{proof}

\doubleQBoundedServiceDifference*

\begin{proof}
    \begin{itemize}
        \item From Theorem~\ref{theorem:bounded-service-difference}, the service bound for each worker is: $2 \cdot (U + Q^u)$.
        \item Since if a client is backlogged, it will have a request and hence be backlogged in all workers. This is because from Line~\ref{line:select_worker}, requests will be distributed to all workers and credit for each worker is exhausted, before replenishing the credits for all workers. 
        \item Therefore, the service bound for \distname is $2 \cdot \vert W \vert (U + Q^u)$.
    \end{itemize}

\end{proof}

\doubleQBacklog*

\begin{proof}

\begin{itemize}

    \item $f$ is continuously backlogged in all workers.
    \item $g$ is not continuously backlogged in at least one worker.
    \item From Lemma~\ref{theorem:backlog}, the service bound is $\vert W \vert \cdot (2 U + 2 Q^u)$ between backlogged and either backlogged or non-backlogged clients. 
    
\end{itemize}

\end{proof}

\begin{restatable}
[\textbf{Latency Bound}]{theorem}{doubleQLatencyBound} Let $A(r)$ and $D(r)$ denote the arrival time and dispatch time of a request $r$. Assume there are in total $n$ clients, $\forall t_1, t_2$, if at $t_1$, a client $f$ is not backlogged and has no requests in the running batch, then the next request $r_f$ with $t_1  < A(r_f) < t_2$ will have its response time bounded: $D(r_f) - A(r_f) \leq (n-1) \vert W \vert \cdot \frac{2U + 2Q^u}{a}$, where $a$ is the lower bound of the server capacity.
\end{restatable}
\label{sec:Doubleq bound}


\begin{proof}
    \begin{itemize}
        \item Since there is no running batch of $f$ in the system, $r_f$ will be selected for the next request for $f$. 
        \item Earlier, we have shown that the bound between a backlogged client and a non-backlogged client in \distname to be $\max_{i} W_i - \min_{i} W_i \leq (2U + 2Q^u) \vert W \vert$. 
        \item Therefore the maximum number of tokens served before $f$ is served again is: $(n-1) \cdot (2U + 2Q^u) \vert W \vert$, where $n - 1$ is the $n - 1$ other client.  
        \item Given that $a$ is the lower bound of the server capacity, the dispatch time for $f$ is therefore bounded: $D(r_f) - A(r_f) \leq (n-1) \vert W \vert \cdot \frac{2U + 2Q^u}{a}$.
    \end{itemize}
\end{proof}

\begin{restatable}[\textbf{Infinite $Q_w$ is not fair}]{theorem}{infiniteQW}
Consider any execution of the \distname Scheduling scheme in which client $Q_w$ is infinite. Such scheduling scheme is not fair. 
\label{theorem:infiniteQ}
\end{restatable}

\begin{proof}
    \begin{itemize}
        \item The requests will not be perfectly load balanced to all workers.
        \item Proof by counterexample: client $f$ sends requests with large prefix matching. Requests from that client will be sent to the same worker hosting the prefix. 
        \item Another client $g$ sends requests with zero prefix matching, the requests will be load-balanced to all workers because of Line~\ref{line:routeEven}.
        \item Client $g$ will be able to get unboundedly more service compared to $f$ as it is replenished more, due to being scheduled to more workers, despite both being backlogged. 
    
    \end{itemize}
\end{proof}

\subsection{Per-Client Round-Robin Can Achieve Fairness. }
\label{sec:DRR client}
\begin{theorem}[\textbf{Service between backlogged or non-backlogged clients is unbounded}] 
\label{per-client-round-robin}
In any interval $[t_1, t_2)$, The difference between the maximum service among all backlogged clients and the minimum service among all backlogged or non-backlogged clients is bounded by a constant independent of the time interval $t_2 - t_1$. 
\end{theorem}

\begin{proof}
\begin{itemize}
    \item The client requests will be load-balanced to all workers. 
    \item Therefore, when a client is backlogged, it is backlogged on all workers. 
    \item We can apply the bound derived in \distname, multiplied by the number of workers $\vert W \vert$, similar to \S\ref{section:doubleq}. 
\end{itemize}
\end{proof}

\end{document}